\documentclass[twoside,leqno,twocolumn]{article}

\usepackage[letterpaper]{geometry}
\usepackage{ltexpprt}
\usepackage{graphicx}
\usepackage{booktabs}
\usepackage{multirow}
\usepackage{xcolor}

\usepackage{todonotes}
\usepackage{multirow}
\usepackage{amsmath}
\usepackage{amssymb}
\usepackage{url}
\usepackage{subcaption}
\usepackage{cleveref}
\usepackage{enumitem}
\usepackage{balance}
\setlist{nolistsep}
\newcommand{\sz}{\scriptstyle }
\newtheorem{definition}{Definition}
\providecommand{\keywords}[1]{\small\textbf{{Keywords---}}#1}
\title{Spatio-Temporal Top-k Similarity Search for Trajectories in Graphs}
\author{
Anne Driemel\\
Petra Mutzel\\
Lutz Oettershagen
}
\author{
    Lutz Oettershagen\thanks{Institute of Computer Science, University of Bonn, Germany, \{lutz.oettershagen,petra.mutzel\}@cs.uni-bonn.de}
    \and Anne Driemel\thanks{Hausdorff Center for Mathematics, University of Bonn, Germany, driemel@cs.uni-bonn.de}
    \and Petra Mutzel\footnotemark[1]
}
\begin{document}
\date{}
\maketitle

\begin{abstract}
We study the problem of finding the $k$ most similar trajectories to a given query trajectory.
Our work is inspired by the work of Grossi et al.~\cite{grossi2020finding} that considers trajectories as
walks in a graph. Each visited vertex is accompanied by a time-interval.
Grossi et al.~define a similarity function that captures temporal and spatial aspects.
We improve this similarity function to derive a new spatio-temporal distance function for which we can show that a specific type of triangle inequality is satisfied. 
This distance function is the basis for our index structures, which can be constructed efficiently, need only linear memory, and can quickly answer queries for the top-$k$ most similar trajectories. 
Our evaluation on real-world and synthetic data sets shows that our algorithms outperform the baselines with respect to indexing time by several orders of magnitude while achieving similar or better query time and quality of results.
\\\keywords{Trajectories, Indexing, Top-$k$ Query}
\end{abstract}

\section{Introduction}
More and more trajectory data is collected due to the ubiquitous availability of sensors and personal mobile devices that allow tracking of movement over time. 
Therefore, trajectory data mining is attracting increasing attention in the scientific literature~\cite{AgrawalFSS93,ChenOO05,ChenSYZ19,ChenSZ+10,HwangKL06,ShangDZJKZ14,tiakas2006trajectory,TiakasR15,XiaWZ+2011}. For any fundamental task in trajectory data mining, the choice of similarity measure is a crucial step in the design process.
Often there are spatial restrictions to the movement and the trajectories of interest are related to a graph, or are mapped to a spatial network.   
We are interested in similarity, which takes such spatial as well as \emph{temporal} aspects into account.
We consider two trajectories as similar if they visit the same or proximate vertices during the same periods of time. 
Our work is inspired by Grossi et al.~\cite{grossi2020finding}, who define a similarity function for two trajectories in a graph. The trajectories can be of different length and the similarity function takes spatial and temporal similarity into account. It can be computed in linear time with respect to the length of the trajectories.
We consider trajectories as sequence of vertices in a graph and for each visited vertex there is a discrete time-interval for the time the trajectory stays at the vertex.
See \Cref{fig:example1v2} for an example.
Based on Grossi et al.~\cite{grossi2020finding}, we introduce an improved and new spatio-temporal similarity and a corresponding distance function for trajectories in graphs. We show that a specific kind of triangle-inequality holds for the distance function under reasonable assumptions.
This distance function provides the basis for new index data structures that allow efficient top-$k$ similarity queries.
A top-$k$ trajectory query $(Q,s)$ specifies a trajectory $Q$ and a time interval $s$. Given a set of trajectories $\mathcal{T}$, 
the result of a top-$k$ query consists of the subset of $\mathcal{T}$ containing all trajectories that have one of the $k$ highest similarities to $Q$ with respect to the time interval $s$. 
These queries have important real-life applications:
\begin{figure}
    \centering
    \includegraphics[width=0.79\linewidth]{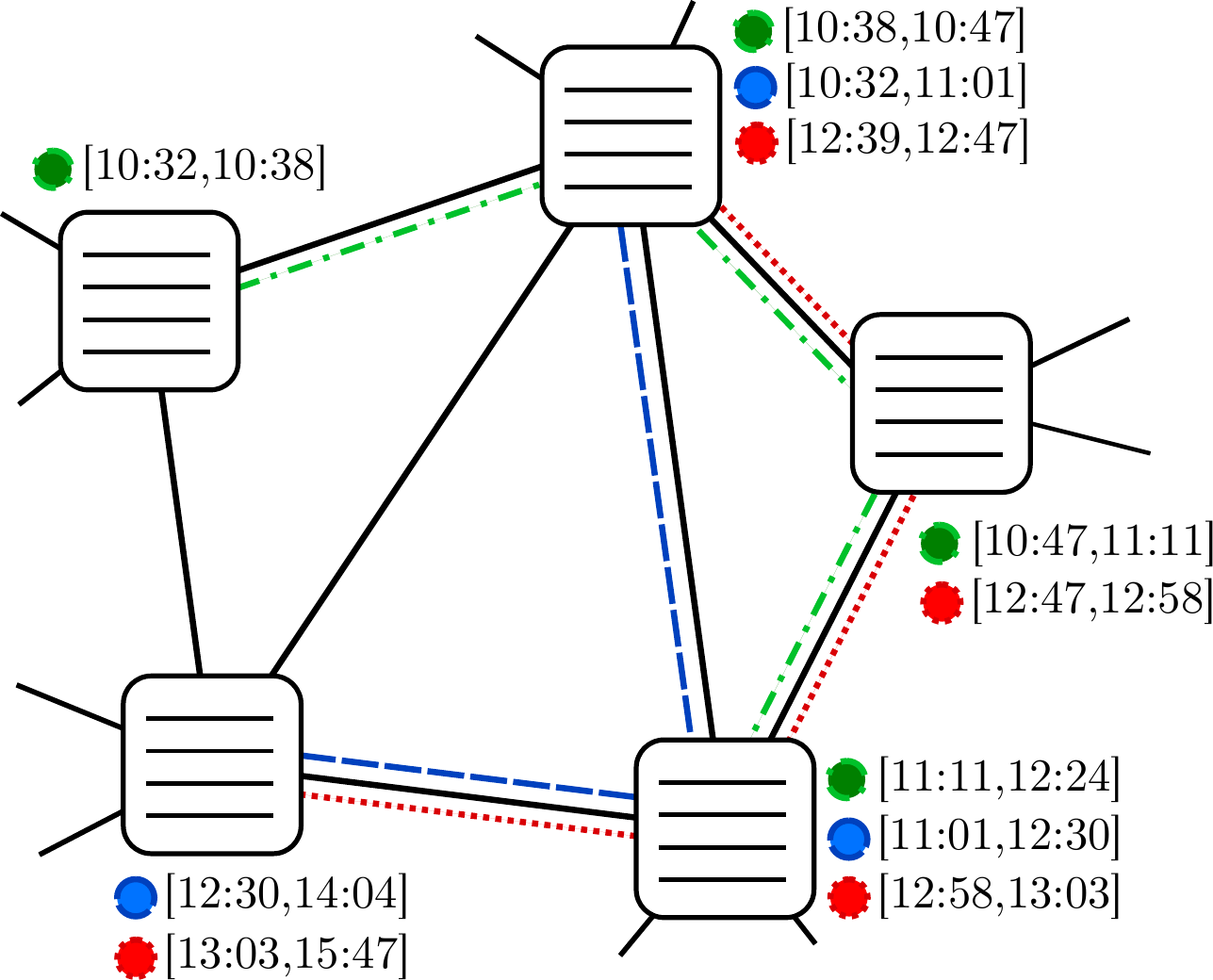}
    \caption{Example for trajectories with time intervals in a network, e.g., an online social network. The trajectories reveal the user behavior, i.e., the times a user visits and leaves a website.}
    \label{fig:example1v2}
\end{figure}

\begin{itemize}
    \item \textbf{Web analytics:} Users of a online social network or web community following links and visiting user pages. 
    The goal is to find similar browsing behavior. \Cref{fig:example1v2} shows an example.  
    \item \textbf{Travel recommendation:} Tourism is one of the largest industries and the emergence of travel focused social networks enables users to share their tours.
    The locations are points-of-interests (POI), and the intervals are the duration person stays at a POI. A query is a request for a recommendation.
    \item \textbf{Animal behavior:} Consider wildlife that is tracked using GPS. The living space of the animal is divided into zones. The goal is to identify similarities in animal behaviors. 
    Vertices represent either specific locations like waterhole or feeding place or territories of animals. 
    \item \textbf{Traffic and crowd analysis:} The goal is to identify person or vehicle flows at specific times through predefined areas.
    Vertices represent these areas. This also includes the application in contact tracing, where we need to determine contacts of an, e.g. infected individual, to other persons. 
\end{itemize}
These applications have in common that one is interested in finding a set of the most similar trajectories to a given one.
This is a fundamental problem in trajectory mining like clustering, outlier detection, classification, or prediction tasks.
It is necessary to select or define an adequate similarity measure or distance function, respectively, that fits the requirements of the application.

\vspace{2mm}
\noindent\textbf{Contributions:}
\begin{enumerate}\itemsep0em
    \item We introduce a spatio-temporal similarity function and show that the triangle-inequality holds under certain conditions for the corresponding distance function. The similarity computation for two trajectories only needs linear time with respect to the length of the longer trajectory.
    \item We design indices that can be constructed very efficiently and use linear memory with respect to the number of trajectories. The indices are based on spatial as well on temporal filtering and allow heuristic top-$k$ similarity queries with short running times and high quality of the results. 
    Additionally, we apply upper bounding, which allows a direct, highly efficient query even without the need for a preprocessed index data structure. In the latter case, the output is exact.
    \item We evaluate our new algorithms on real-world and synthetic data sets. Our new solutions outperform the baselines (including \cite{grossi2020finding}) with respect to indexing time by several orders of magnitude. Moreover, our query times are substantially faster, and the quality of the results is better than or on par with the baselines algorithms.  
\end{enumerate}
\subsection{Related Work}
Since trajectory similarity is of high interest for many data analytics tasks, many different similarity measures 
have been used, e.g., based on dynamic time warping, Euclidean distances, or edit distances. Su et al.~\cite{distancesurvey} provides a nice overview.
For trajectory analysis in networks, many approaches have concentrated on the spatial similarity only, and a few consider spatio-temporal similarity.
Hwang et al.~\cite{HwangKL06} have suggested a similarity measure based on the network distance
measuring spatial and temporal similarity. However, a set of nodes need to be selected in advance, and spatial similarity then means passing through the same nodes simultaneously.
Xia et al.~\cite{XiaWZ+2011} use a similarity measure for network constrained trajectories based on an extension of the Jaccard similarity.
As a similarity measure, they use the product of spatial and temporal similarity.
Tiakas et al.~\cite{tiakas2006trajectory,TiakasR15} suggest a weighted sum of spatial and temporal similarity. 
Their similarity function works for two trajectories with the same length and can be computed in linear time with respect to the length of the given trajectories.
Shang et al.~\cite{simjoin} also use a weighted sum of spatial and temporal similarity for similarity-joins of trajectories.

Another way to approach the problem is to use a distance measure based on the discrete Fr\'echet distance, or dynamic time warping, which optimize over all vertex-mappings between the two trajectories that respect the time-ordering, where the underlying metric would be derived from the shortest-path metric given by the graph. Near-neighbor data structures have been studied theoretically with specific conditions on the underlying graph and the length of the queries, see \cite{driemelsublinear2019, Indyk02}.

Our work is inspired by Grossi et al.~\cite{grossi2020finding}. They suggest a spatio-temporal similarity measure for two trajectories in a graph. The trajectories can be of differend length and if the pairwise distances are given, then the measure can be computed in linear time with respect to the length of the trajectories. The authors also suggest an algorithm for answering the top-$k$ trajectory query problem. For speeding up the computations, they suggest an indexing method based on interval trees and a method to approximate their similarity measure.
We provide a more detailed description of their work and a comparison to our approach in \Cref{sec:grossi}.
\section{Preliminaries}
An \emph{undirected} and weighted \emph{graph} $G=(V, E, c)$ consists of a finite set of vertices $V$, a finite set $E\subseteq\{\{u,v\}\subseteq V\mid u\neq v\}$ of undirected edges and a cost function $c:E\rightarrow \mathbb{R}_{>0}$ that assigns a positive cost to each edge $e\in E$.
A \emph{walk} in $G$ is an alternating sequence of vertices and edges connecting consecutive vertices. 
A \emph{path} is a walk that visits each vertex at most once. The cost of a walk or path is the sum of its edge costs.
Let $d(u,v)$ denote the shortest path distance between $u,v\in V$.

\begin{definition}[Trajectory]
Let $G=(V,E,c)$ be an undirected, weighted and connected graph. A \emph{trajectory} $T$ is a sequence of pairs $((v_1,t_1),\ldots,(v_\ell,t_\ell))$, such that for $1\leq i \leq \ell$ the pair $(v_i, t_i)$ consists of $v_i\in V$ and a discrete time interval $t_i=[a_i,b_i]$ with $a_i,b_i\in \mathbb{Z}$, $a_i< b_i$ and $a_{i+1}=b_i$ for $1\leq i<\ell$.
\end{definition}
The \emph{starting time} of $T$ is $T.start=a_1$ and the \emph{ending time} $T.end=b_\ell$.
We denote with $\mathcal{I}(T)$ the total interval in which trajectory $T$ exists, i.e., from $T.start$ to $T.end$. %
For a trajectory $T$ and a time interval $t=[a,b]$ we define $T[t]$ as the \emph{time-restricted} trajectory that is intersected with $t$, i.e., $T[t]=((v_i,t'_i),\ldots,(v_j,t'_j))$ with $v_i$ ($v_j$) being the first (last) vertex of $T$ such that for $t_i=[a_i,b_i]$ it holds that $b_i>a$ (and for $t_j=[a_j,b_j]$ $a_j<b$, resp.), $t'_i=\max\{t_i,a\}$ and $t'_j=\min\{t_j,b\}$.
We assume for $T=((v_1,t_1),\ldots,(v_\ell,t_\ell))$ that $v_i\neq v_{i+1}$ for all $1\leq i <\ell$.
We say trajectory $T$ \emph{intersects} a time interval $t$ if there is a $(v_i, t_i)\in T$ with $t_i\cap t\neq \emptyset$.
\section{Spatio-Temporal Similarity}
We define our new similarity function for trajectories on networks. %
The goal of is to capture both temporal and spatial aspects, such that if two trajectories are often in close proximity, i.e., visiting vertices that are close to each other during the same period of time, then the similarity should be high.
\begin{definition}\label{def:sim}
Let $T=((v_1,t_1),\ldots,(v_\ell,t_\ell))$ and $Q=((u_1,s_1),\ldots,(u_k,s_k))$ be two trajectories, and $s$ a time interval. We define the similarity of $Q$ and $T$ in the time interval $s$ as
\begin{align*}
Sim(Q,T,s)=\frac{1}{|s|}\cdot\sum_{\substack{(v_i,t_i)\in T\\(u_j,s_j)\in Q}}{|s\cap t_i\cap s_j|\cdot e^{-d(v_i,u_j)}}\text{.}
\end{align*}
\end{definition}
Notice that for two trajectories $T$ and $Q$, and a time interval $s$, it holds that $0\leq Sim(Q,T,s) \leq 1$.  $Sim(Q,T,s)$ is minimal if the common intersection of the time intervals is empty. In this case $Sim(Q,T,s)=0$. 
\begin{lemma}\label{lemma:props1}
    Let $Q$ and $T$ be trajectories and $s$ a time interval with $s\subseteq \mathcal{I}(Q)$. It holds that
    \begin{enumerate}
        \item $Sim(Q,T,s) = Sim(T,Q,s)$, and
        \item $Q[s]=T[s]$ if and only if $Sim(Q,T,s) = 1$.
    \end{enumerate}
\end{lemma}
\begin{proof}
    (1.) The shortest path metric is symmetric, i.e., $d(u,v)=d(v,u)$ for all $u,v\in V$. The summation is over the same pairs of $(v_i,t_i)\in T$ and $(u_j,s_j)\in Q$, and the intersection of the intervals is commutative. Therefore, the result follows.
    
    (2.) $\Rightarrow:$ Notice that if $Q[s]=T[s]$ in each step of the summation $e^{-d(u,v)}=1$. Because $s\subseteq \mathcal{I}(Q)$, the result of the summation is $s$ and normalization is $1$. 
    \\$\Leftarrow:$ %
    Assume that $Sim(Q,T,s) = 1$ but $Q[s]\neq T[s]$, i.e., $Q[s]$ and $T[s]$ differ in the vertices they visit or the times when they visit them. In the first case, due to the strictly positive edge weights, there is a vertex pair such that $e^{-d(u,v)}<1$, however for all other vertex pairs the value $e^{-d(u',v')}$ is at most $1$. Because the intervals are intersected with the interval $s$ the total sum will be less than $|s|$ and leads to a contradiction to the assumption. Analogously, in the case that $\mathcal{I}(T[s])<|s|$ a contradiction follows.
    Now, the case that $Q[s]$ and $T[s]$ differ in the times when they visit the vertices. 
    Because of the assumption that a trajectory does not stay at the same vertex in two consecutive time intervals, there is an intersection of time intervals in which $Q[s]$ and $T[s]$ visit different vertices $u$ and $v$. Due to the strictly positive edge weights it is $e^{-d(u,v)}<1$. This leads again to a contradiction.
\end{proof}
For the computation of the similarity, the shortest-path distances between the vertices of the graph is needed. These distances can be precomputed for all vertices or computed on-the-fly for vertices $u$ that are visited by the query trajectory $Q$. %
\begin{theorem}\label{theorem:runningtime}
    Let $Q$ and $T$ be trajectories, and $s$ a time interval, the computation of the similarity $Sim(Q,T,s)$ takes $\mathcal{O}(|Q|+|T|)$ time, if the shortest path distance $d(u,v)$ between $u,v\in V$ can be obtained in constant time.
\end{theorem}
\begin{proof}
    Consider the query trajectory $Q=((u_1,s_1),\ldots,(u_i,[a_i,b_i]),\ldots,(u_k,s_k))$ and the trajectory $T=((v_1,t_1),\ldots,(v_j,[c_j,d_j]),\ldots,(v_\ell,t_\ell))$. 
    We start the computation with $i=1$ and $j=1$, and $|s\cap t_1\cap s_1|$ is either zero or larger than zero. In the first case we can increase both $i$ and $j$. In the second case, we increase $i$ if $b_i<d_j$ or $j=\ell$, and we increase $j$ if $b_i>d_j$ or $i=k$. We repeat this for maximal $|Q|+|T|$ times and find all pairs $(u_i,s_i)$ and $(v_j,t_j)$ that have non-empty intersection.
\end{proof}
\Cref{def:sim} is similar to the similarity function defined in~\cite{grossi2020finding}, however our improvements allow to prove useful properties for the corresponding distance function.
We now define the distance function based on the similarity and show a specific type of triangle inequality.
\begin{definition}\label{def:distance}
    Let $Q$ and $T$ be trajectories, and $s$ a time interval.
    We define the distance $Dist(Q,T,s)=1-Sim(Q,T,s)$.
\end{definition}
\begin{lemma}\label{lemma:triangle:ineq}
    Let $Q$, $T$ and $R$ be trajectories, and $s$ a time interval. If $s\subseteq \mathcal{I}(Q)$,
    then $Dist(Q,T,s)\leq Dist(Q,R,s)+Dist(R,T,s)$.
\end{lemma}

\begin{proof}
    Let $t = \mathcal{I}(T)$ and $r = \mathcal{I}(R)$.
    We can assume without loss of generality that $s = \mathcal{I}(Q)$.
    We show that $1-Sim(Q,T,s) \leq 1-Sim(Q,R,s) + 1-Sim(R,T,s)$. This is equivalent to
    $ 1+Sim(Q,T,s) \geq Sim(Q,R,s)    + Sim(R,T,s)$.
    By substituting \Cref{def:sim} and using the fact that
    \begin{align*}
     |s \cap t| = \sum_{\substack{(u_j,s_j)\in Q\\(v_i,t_i)\in T}} |s \cap s_j \cap t_i |,
    \end{align*}
    we can rewrite the above equivalently as
    \begin{align*}
    &|s|-|s \cap t|+\sum_{\substack{(u_j,s_j)\in Q\\(v_i,t_i)\in T}}{|s\cap s_j\cap t_i|\cdot (1+e^{-d(u_j,v_i)})} \geq \\&\hspace{2.3cm}\sum_{\substack{(u_j,s_j)\in Q\\(w_k,r_k)\in R}}{|s\cap s_j\cap r_k|\cdot e^{-d(u_j,w_k)}} \\&\hspace{2.2cm}+ \sum_{\substack{(w_k,r_k)\in R\\(v_i,t_i)\in T}}{|s\cap r_k\cap t_i|\cdot e^{-d(w_k,v_i)}}\text{.}
    \end{align*}
    We now want to show that the above inequality holds.
    Consider the following multisets of vertex pairs.
    $A$ contains the pairs $(u_j,v_i)$ that are summed up on the left side of the inequality for which $|s\cap s_j\cap t_i|>0$, where each $(u_j,v_i)$ is in $A$ exactly $|s\cap s_j\cap t_i|$ times.
    Similarly, $B$ contains the pairs $(u_j,w_k)$ that are summed up during the first summation on the right-hand side of the inequality for which $|s\cap s_j\cap r_k|>0$, where each $(u_j,w_k)$ is in $B$ exactly $|s\cap s_j\cap r_k|$ times.
    And finally, $C$ contains the pairs $(w_k,v_i)$ that are summed up during the second summation on the right-hand side of the inequality for which $|s\cap r_k\cap t_i|>0$, where each $(w_k,v_i)$ is in $C$ exactly $|s\cap r_k\cap t_i|$ times.
    Then, we show 
    \begin{align}\label{eq:multiset}
        &|s|-|s \cap t|+\sum_{(u_j,v_i)\in A} (1+e^{-d(u_j,v_i)}) \geq \nonumber\\&\hspace{1cm}\sum_{(u_j,w_k)\in B}e^{-d(u_j,w_k)}+ \sum_{(w_k,v_i)\in C}e^{-d(w_k,v_i)}\text{.}
    \end{align}
    We show that the multisets $A$, $B$ and $C$ contain vertex pairs such that the inequality holds.
    And let $p\subseteq s$ be an interval of length one. For each possible $p$ we may have some vertex pairs in the multisets.  

    We need the consider the following cases: %
    \begin{enumerate}\itemsep0em
         \item $p\cap t \neq \emptyset$ and $p\cap r = \emptyset$: $A$ contains vertex pairs $(u,v)$ but neither $B$ nor $C$ contain corresponding pairs. Therefore, favoring the left side of \cref{eq:multiset}.
         \item $p\cap t \neq \emptyset$ and $p\cap r \neq \emptyset$: 
         There are $(v_i, u_j)\in A$, $(v_i, w_k)\in B$ and $(w_k, u_j)\in C$.
         In this case it holds that $1+e^{-d(u_j,v_i)}\geq e^{-d(u_j,w_k)}+e^{-d(w_k,v_i)}$.
         \item $p\cap t = \emptyset$ and $p\cap r \neq \emptyset$: There are no corresponding vertex pairs in $A$ and $C$ but in $B$. However, this can only be the case for $|s|-|s \cap t|$ pairs and each contributes at most $1$ to the right-hand side.  
          \item $p\cap t=\emptyset$ and $p\cap r = \emptyset$: There are no corresponding vertex pairs in $A$, $B$ or $C$.
    \end{enumerate}
\end{proof}
Now, we show a strong relationship between the similarities, or distances, of two trajectories with respect to two different time intervals.
\begin{lemma}\label{lemma:bound}
    Let $Q$ and $T$ be trajectories, and $s$ and $t$ time intervals with $\mathcal{I}(Q)=s\subseteq t$.
It holds that %
$Dist(Q,T,t)=1-\frac{|s|}{|t|}+\frac{|s|}{|t|}Dist(Q,T,s)$.
\end{lemma}
\begin{proof}
   Assuming $s \subseteq t$  and using \Cref{def:sim} it holds that
   \begin{align*}
      Sim(Q,T,t) &= \frac{1}{|t|} \sum_{\substack{(u_j,s_j)\in Q\\(v_i,t_i)\in T}}{|t\cap s_j\cap t_i|\cdot e^{-d(u_j,v_i)}} \\
      & = \frac{1}{|t|} \sum_{\substack{(u_j,s_j)\in Q\\(v_i,t_i)\in T}} {|s \cap s_j\cap t_i|\cdot e^{-d(u_j,v_i)}}
    \end{align*}
   since $s_j \subseteq s \subseteq t$ for all $s_j$.
   Now we can apply \Cref{def:sim} again and obtain
   \begin{align*}
      Sim(Q,T,t) & = \frac{|s|}{|t|} Sim(Q,T,s).  %
    \end{align*}
      Finally, applying \Cref{def:distance} leads to the result.
\end{proof}
\Cref{lemma:triangle:ineq} and \Cref{lemma:bound} are the basis for our indices that we present in the following section.
\section{Indexing the Trajectories}
We introduce efficient indexing methods for the trajectories by applying temporal and spatial filters.
First, we give a high-level view of our approach, which consists of two phases: 1. An offline phase for preparing the index. Given a set of trajectories $\mathcal{T}$ a preprocessing phase constructs the index $\mathcal{D}$ that allows efficient queries. 
2. The query phase. %
Given a query trajectory $Q$ and a time interval $s$,
the index $\mathcal{D}$ first determines a candidate set $\mathcal{C}\subseteq \mathcal{T}$. 
For each $T\in \mathcal{C}$ the query algorithm computes the similarity $Sim(Q,T,s)$, and keeps all trajectories with a top-$k$ similarity in a heap data structure. 
The query result is the set of all trajectories with a top-$k$ similarity to $Q$ w.r.t.~$s$.
Notice that the candidate set may contain all trajectories stored in $\mathcal{D}$, e.g., if $k\geq|\mathcal{T}|$ or if all trajectories have the same similarity to the query. In the following, we describe the techniques that achieve small candidate sets wherever possible. Our techniques are based on filters for the temporal and the spatial domain.
Our indexing algorithms, as well as our query algorihms, are embarrassingly parallel.

\subsection{Pivot-Based Spatial Filters}\label{sec:pivot_filter}
We choose $h\in\mathbb{N}$ vertices $p_1,\ldots,p_h$ from which we construct $h$ pivot trajectories $P_1,\ldots,P_h$. 
The $h$ vertices are the ones that are most-frequently visited by trajectories, where we also count multiple visits from a trajectory $T$ at a vertex.
Each pivot trajectory $P_i$ stays stationary at vertex $p_i$ during the time interval $t=[a,b]$, where $a$ is the earliest starting and $b$ the latest ending time over all trajectories $T\in \mathcal{T}$.
Next we compute the pairwise distances $Dist(T,P_i,t)$ between all $T\in\mathcal{T}$ and $P_i$ for $1\leq i \leq h$ and store these distances together with the pivot trajectories. 
Given a query $(Q,s)$ we compute $Dist(Q,P_i,t)$ for $1\leq i \leq h$.
Using \Cref{lemma:triangle:ineq} and \Cref{lemma:bound}, it follows that  %
\begin{align*}
    |Dist(Q,P_i,t)-Dist(P_i,T,t)| &\leq Dist(Q,T,t)\\
&\hspace{-1cm}=  1-\frac{|s|}{|t|}+\frac{|s|}{|t|}Dist(Q,T,s), 
\end{align*}
where we use that $t \subseteq \mathcal{I}(P_i)$ which holds by construction of $P_i$.
We can use the above bound to filter out a lot of trajectories from the candidate set that are too far away from the query to be in the top-$k$ result set.  
To this end, we use a threshold radius $r$ such that we only keep trajectories $T$ for
which
\[
|Dist(Q,P_i,t)-Dist(T,P_i,t)| \leq r
\]
for all pivots $P_i$ with $1\leq i \leq h$. 

The running time needed for filtering the trajectories during a query is in $\mathcal{O}(|\mathcal{T}|+h\cdot |Q|)$.
The construction of the index utilizing pivot-based spatial filtering is efficient---we only need not compute the distance between the $h$ pivot trajectories and all $T\in\mathcal{T}$, each in $\mathcal{O}(|T|+|P_i|)$.
\begin{theorem}\label{theorem:pivot_rt_mem}
    The index based on pivot-based spatial filters can be computed in $\mathcal{O}(|\mathcal{T}|\cdot h m)$ time, where $m$ is the the maximal length of a trajectory over $\mathcal{T}$.
    The memory needed for storage is in $\mathcal{O}(|\mathcal{T}|\cdot h)$.
\end{theorem}
\subsection{Temporal Filter}\label{sec:temp_filter}
Notice that trajectories that have empty intersection with the query interval $s$ do not have to be considered in the candidate set 
$\mathcal{C}\subseteq\mathcal{T}$.
To filter out such trajectories we construct a binary interval tree using the following procedure.
For each node $h$ in the tree, we have a set of trajectories $\mathcal{T}_h$. We compute the median $m$ of the end-points in $\mathcal{T}_h$ and assign all trajectories that end before $m$ to the left child and all trajectories that start after $m$ to the right child of $h$. All other trajectories are stored at $h$. 
We proceed recursively until we reach a minimum size for the trajectory set $\mathcal{T}_{l}$, where $l$ is a leaf of the tree.

We combine the temporal filter with the pivot-based spatial filter by using a pivot-based spatial filter at each node $h$ for the trajectories stored at node $h$.

The running time needed for temporal filtering during a query is in $\mathcal{O}(|\mathcal{C}|+h\cdot |Q|)$, with $|\mathcal{C}|$ being the size of the candidate set returned by the index. 
During the construction, we have to do the pivot based filter construction at each vertex.
\begin{theorem}\label{theorem:temporal_rt_mem}
    The tree index can be computed in $\mathcal{O}(\log(|\mathcal{T}|)\cdot |\mathcal{T}|\cdot h m)$ time, where $m$ is the the maximal length of a trajectory in $\mathcal{T}$.
    The memory needed for storage is in $\mathcal{O}(|\mathcal{T}|\cdot h)$.
\end{theorem}
\subsection{Upper Bounding}\label{sec:upperbounding}
During the computations of the similarities between $Q$ and a trajectory $T$ in a set of trajectories $\mathcal{T}$ we can apply the following upper bounding technique. 
Let $T_1,\ldots,T_{|\mathcal{C}|}$ be the trajectories of the candidate set in order of processing. 
After computing the similarity of the first $k$ trajectories, we can stop the similarity computation between $Q$ and $T_h$ for $h>k$ early if we can assure that $Sim(Q,T_h,s)$ is smaller than any similarity between $Q$ and any $T\in \mathcal{T}$ computed so far. To this end, we iteratively update an upper bound $\bar{s}$ for the value of $Sim(Q,T_h,s)$.
Consider the computation of the similarity $Sim(Q,T_{h},s)$ described in the proof of Theorem \ref{theorem:runningtime}. 
At each step, before increasing $i$ or $j$, we obtain the upper bound $\bar{s}$ for $Sim(Q,T_{h},s)$ by assuming that in each remaining time step the trajectories are at the same vertices. If $\bar{s}$ is smaller than the $k$ lowest similarity found so far, we stop the computation of $Sim(Q,T_{h},s)$ and proceed with $Sim(Q,T_{h+1},s)$. 
\section{Comparison to Existing Algorithms}\label{sec:grossi}
Grossi et al.~\cite{grossi2020finding} introduce three algorithms for answering top-$k$ similarity queries in a spatio-temporal setting. The idea of their baseline algorithm is to have a preprocessing phase that constructs an interval tree at each vertex $v$ of the graph.
The interval tree at $v\in V$ contains all pairs of $(T.id,t)$ if trajectory $T\in \mathcal{T}$ visits $v$ or any of its adjacent vertices during time interval $t$. Here, $T.id$ is the identifier of the trajectory $T$.
Then, using the constructed index, a query $(Q,s)$ is answered by visiting all vertices $v$ with $(v,t)\in Q$ and collecting all ids of trajectories that visit vertex $v$ or any of its neighbors during $t\cap s$.
With the collected set of ids, the candidate set of trajectories can be evaluated, and the top-$k$ similar trajectories are found. 
Therefore, the running time and memory requirements depend on the number of trajectories, the lengths of the trajectories, and the vertex degrees.
Moreover, the algorithm solves a special case, in which only trajectories are considered that have at least one vertex in hop-distance less or equal to $1$ to a vertex of the query trajectory. 
A simple example for which the algorithm fails to find a similar trajectory can be constructed in a graph consisting of a chain of four vertices, i.e., $G=(\{v_1,\ldots,v_4\},\{v_1v_2,v_2v_3,v_3v_4\})$, and trajectories $T=((v_1,[0,1]))$ and $Q=((v_4,[0,1]))$. 
After the preprocessing phase, only the interval trees at the vertices $v_1$ and $v_2$ contain the id of $T$. For a query $(Q,[0,1])$, the algorithm will only look at the empty interval tree at $v_4$ and cannot find $T$.
Grossi et al.~\cite{grossi2020finding} also introduce two heuristic algorithms for the top-$k$ query problem. Their idea is to reduce the graph size and then shrink the length of the query trajectory or all trajectories to save running time by reducing the number of distance computations between vertices. However, this can also lead to larger candidate sets, and hence more evaluations of the similarity function are necessary.

Our algorithms differ from the ones suggested by Grossi et al.~\cite{grossi2020finding} as follows. First, we introduced an alternative and improved similarity function for which we showed certain metric-like properties. Our indices use these properties to reduce the size of the trajectory candidate set and, hence, reduce the number of similarity computations. The memory requirements of our indices are independent of the size of the graph and only linear in number of the trajectories (see Theorems \ref{theorem:pivot_rt_mem} and \ref{theorem:temporal_rt_mem}).
By using the upper bounding technique without preprocessing and constructing an index, we obtain an exact algorithm that is competitive in terms of running time.
\section{Experiments}
In this section, we evaluate our new algorithms and compare them to the approaches suggested in~\cite{grossi2020finding}.
We are interested in answering the following questions:
\begin{itemize}\itemsep0em
    \item[\textbf{Q1:}] How fast are the indexing times of our algorithms compared to the algorithms in~\cite{grossi2020finding}? %
    \item[\textbf{Q2:}] How fast are queries of our algorithms compared to the baseline and to the heuristics in~\cite{grossi2020finding}? Do our index solutions improve the query times?
    \item[\textbf{Q3:}] How good is the quality of the approximated top-$k$ queries?
    \item[\textbf{Q4:}] How do the choices of the radius $r$ and the number of pivots $h$ impact running time and accuracy?
    \item[\textbf{Q5:}] How much does the upper bounding improve the running time?
\end{itemize}
\subsection{Algorithms and Experimental Protocol}
We implemented the following new algorithms:
\begin{itemize}\itemsep0em
    \item \textsc{Exact} is the linear scan over the complete data set that does not use indexing. %
    \item \textsc{Tree} is the index that uses an interval tree with additional pivot-based spatial filtering at each node of the interval tree (see \Cref{sec:temp_filter}). 
    \item \textsc{Pivot} is the index that applies the pivot-based spatial filtering globally (see \Cref{sec:pivot_filter}). 
\end{itemize}
\textsc{Exact}, \textsc{Tree} and \textsc{Pivot} use the upper bounding technique (\Cref{sec:upperbounding}).
Furthermore, we implemented the following algorithms from~\cite{grossi2020finding}:
\begin{itemize}\itemsep0em
    \item \textsc{Gbase,} the baseline algorithm in~\cite{grossi2020finding} (see~\Cref{sec:grossi}).
    \item \textsc{Gshq} and \textsc{Gshqt} denote their heuristic algorithms based on shrinking the graph and the trajectories and gaining advantage of the smaller graph size and reduced trajectory lengths (see~\cite{grossi2020finding}).
\end{itemize}
All of our implementations use the similarity measure defined in \Cref{def:sim}.
We implemented all algorithms in C++ using GNU CC Compiler 9.3.0 with the flag \texttt{--O2}.
All experiments were conducted on a workstation with an AMD EPYC 7402P 24-Core Processor with 2.80 GHz and 256 GB of RAM running \text{Ubuntu 18.04.3} LTS. The source code and data sets are online available at \url{https://gitlab.com/tgpublic/topktraj}.
\subsection{Data Sets}
For the evaluation of the algorithms, we used the following data sets:
\begin{table}[t]
    \centering
    \caption{Statistics and properties of the synthetic and real-world data sets.}  %
    \label{table:datasets_stats2}
    \resizebox{1.0\linewidth}{!}{ 	\renewcommand{\arraystretch}{1}
        \hspace{-3mm}
        \begin{tabular}{lcccr}\toprule
            \multirow{3}{2cm}{\vspace*{4pt}\textbf{}\vspace*{4pt}}&\multicolumn{4}{c}{\textbf{Properties}}\\
            \cmidrule{2-5}
            \textbf{Data~set}                 &   $|V|$       & $|E|$         &  $\#$ Traj.  &  $\varnothing$ Traj. Len.  \\\midrule
            \emph{Facebook1} 		   &   $4\,039$    &   $88\,234$   &  $1\,000$     & $1\,482.0\sz\pm301.4$ \\ 
            \emph{Facebook2} 		   &   $4\,039$    &   $88\,234$   &  $10\,000$    & $1\,497.6\sz\pm287.3$ \\ 
            \emph{Milan} 		       &   $3\,000$    &  $123\,406$   &  $9\,525$     & $141.5\sz\pm129.9$  \\  
            \emph{T-Drive} 	           &   $2\,000$    &  $500\,930$   &  $10\,357$    & $598.0\sz\pm457.4$ \\  
            \bottomrule
        \end{tabular}
    }
\end{table}

\begin{itemize}
    \item \emph{Facebook 1\&2:} The network consists of Facebook friendship relations~\cite{leskovec2012learning} and is provided by the \emph{Stanford Network Analyses Project}\footnote{\url{https://snap.stanford.edu/data/ego-Facebook.html}}. We have generated synthetic trajectories. 
    \item \emph{Milan:} The Milan data set is based on GPS trajectories of private cars in the city of Milan\footnote{\url{https://sobigdata.d4science.org/catalogue-sobigdata?path=/dataset/gps_track_milan_italy}}.  
    \item \emph{T-Drive:}  The data set contains GPS data of taxi trajectories in Beijing~\cite{tdrive1,tdrive2}. 
\end{itemize}
For the \emph{Milan} and \emph{T-Drive} data set, we generated a graph by first interpreting each GPS location point as a vertex and then clustering these vertices using the $k$-means algorithm. The resulting clusters are the final vertices. Two clusters are connected by an edge if at least one trajectory visits a vertex in each of both clusters in a consecutive time interval. 
We assign the distance between the centers of the clusters as the distance to the edge.
\Cref{table:datasets_stats2} shows some statistics for the data sets.

\subsection{Results}
We answer questions \textbf{Q1} to \textbf{Q5}.

\medskip
\noindent\textbf{Q1: } \Cref{table:indexing_times} shows the running times for indexing the data sets.
For \textsc{Tree} and \textsc{Pivot}, we choose $h=8$ pivots  for both \emph{Facebook} data sets and set $h=64$ for the \emph{T-Drive} data set. In case of the \emph{Milan} data set, we choose $h=32$ for \textsc{Tree} and $h=16$ for \textsc{Pivot}.
The construction of our index structures is several orders of magnitude faster than that of the algorithms suggested in \cite{grossi2020finding}.  
The largest  speed-up is achieved  for the \emph{T-Drive} data set, for which \textsc{Pivot} is over $22\,000$ times faster. For \emph{Facebook2} \textsc{Pivot} is over $6\,000$ faster. Out of all indexing approaches, as expected, \textsc{Pivot} is the fasted method for all data sets.
\textsc{Tree} is the second fastest with very large speed ups compared to \textsc{Gbase}, \textsc{Gshq} and \textsc{Gshqt}.
The low indexing times allow us to learn the parameter $h$, i.e., finding a suitable number of pivots.
\begin{table}[ht]
    \centering
    \caption{Indexing times in seconds.}  
    \label{table:indexing_times}
    \resizebox{0.85\linewidth}{!}{ 	\renewcommand{\arraystretch}{0.9}
        \hspace{-3mm}
        \begin{tabular}{lccccc}\toprule
            \multirow{5}{1cm}{\vspace*{4pt}\textbf{}\vspace*{4pt}}&\multicolumn{5}{c}{\textbf{Algorithm}}\\
            \cmidrule{2-6}
            \textbf{Data~set}       & \textsc{Tree} & \textsc{Pivot}  &  \textsc{Gbase} &  \textsc{Gshq} &  \textsc{Gshqt}\\ \midrule
            \emph{Facebook1}        &  $0.14$       &    $0.06$       &  $188.23$       &  $32.77$       & $25.76$        \\
            \emph{Facebook2}        &  $1.80$       &    $0.47$       &  $3\,132.22$    &  $485.26$      & $324.47$       \\	
            \emph{Milan}            &  $0.26$       &    $0.07$       &   $232.13$      &  $43.66$       & $34.78$        \\
            \emph{T-Drive}          &  $1.09$       &    $0.49$       &  $11\,823.81$   &  $475.06$      & $403.50$       \\
            \bottomrule
        \end{tabular}
    }
\end{table}
\begin{table}[ht]
    \centering
    \caption{Threshold values $r$ used for the pivot based filters during query time.}  
    \label{table:treshold_radius}
    \resizebox{0.8\linewidth}{!}{ 	\renewcommand{\arraystretch}{0.8}
        \hspace{-3mm}
        \begin{tabular}{lcccc}\toprule
            \multirow{3}{1cm}{\vspace*{4pt}\vspace*{4pt}}&\multicolumn{4}{c}{\textbf{Data~set}}\\
            \cmidrule{2-5}
            \textbf{Index}        & \emph{Facebook1} & \emph{Facebook2}  &  \emph{Milan} &  \emph{T-Drive} \\ \midrule
            \textsc{Tree}         &  $0.1$           &    $0.1$          &  $0.2$        &  $0.2$         \\
            \textsc{Pivot}        &  $0.1$           &    $0.1$          &  $0.02$       &  $0.25$        \\	  
            \bottomrule
        \end{tabular}
    }
\end{table}
\begin{table}[ht]
    \centering
    \caption{Query times in seconds for top-$k$ similarity queries. The running times are the average and standard deviations over $100$ queries. The fastest running time in each row is highlighted.}  
    \label{table:query_times}
    \resizebox{1\linewidth}{!}{ 	\renewcommand{\arraystretch}{0.8}
        \hspace{-3mm}
        \begin{tabular}{lccccccc}\toprule
            \multirow{5}{1cm}{\vspace*{4pt}\textbf{}\vspace*{4pt}}&\multicolumn{7}{c}{\textbf{Algorithm}}\\
            \cmidrule{3-8}
            \textbf{Data set}& $k$    & \textsc{Exact}& \textsc{Tree}& \textsc{Pivot}  &  \textsc{Gbase} &  \textsc{Gshq} &  \textsc{Gshqt}\\ \midrule
            \emph{Facebook1} &  $1$   &  $0.176\sz\pm0.04$&  $0.071\sz\pm0.03$   &   $\textbf{0.066}\sz\pm0.03$  &   $0.235\sz\pm0.05$   &  $0.176\sz\pm0.03$  & $0.168\sz\pm0.04$ \\	
            \emph{Facebook2} &  $1$   &  $1.709\sz\pm0.53$&  $0.673\sz\pm0.40$   &   $\textbf{0.666}\sz\pm0.39$  &   $2.529\sz\pm0.58$   &  $1.910\sz\pm0.48$  & $1.556\sz\pm0.65$ \\	
            \emph{Milan}     &  $1$   &  $\textbf{0.004}\sz\pm0.00$&  $0.014\sz\pm0.01$   &   $0.015\sz\pm0.00$  &   $0.016\sz\pm0.00$   &  $0.116\sz\pm0.05$  & $0.128\sz\pm0.05$ \\
            \emph{T-Drive}   &  $1$   &  $0.019\sz\pm0.02$&  $\textbf{0.017}\sz\pm0.01$   &   $\textbf{0.017}\sz\pm0.01$  &   $2.101\sz\pm0.87$   &  $0.560\sz\pm0.26$  & $0.615\sz\pm0.26$ \\\cmidrule{1-8}	
            \emph{Facebook1} &  $4$   &  $0.179\sz\pm0.04$&  $0.069\sz\pm0.03$   &   $\textbf{0.067}\sz\pm0.03$  &   $0.238\sz\pm0.04$   &  $0.176\sz\pm0.03$  & $0.168\sz\pm0.04$ \\
            \emph{Facebook2} &  $4$   &  $1.735\sz\pm0.53$&  $0.680\sz\pm0.40$   &   $\textbf{0.668}\sz\pm0.40$  &   $2.563\sz\pm0.55$   &  $1.910\sz\pm0.48$  & $1.560\sz\pm0.65$ \\	
            \emph{Milan}     &  $4$   &  $\textbf{0.009}\sz\pm0.01$&  $0.017\sz\pm0.01$   &   $0.017\sz\pm0.01$  &   $0.019\sz\pm0.01$   &  $0.116\sz\pm0.05$  & $0.127\sz\pm0.05$ \\
            \emph{T-Drive}   &  $4$   &  $0.028\sz\pm0.01$&  $0.021\sz\pm0.01$   &   $\textbf{0.020}\sz\pm0.01$  &   $0.234\sz\pm0.09$   &  $0.567\sz\pm0.26$  & $0.615\sz\pm0.26$ \\\cmidrule{1-8}
            \emph{Facebook1} &  $16$  &  $0.183\sz\pm0.04$&  $0.070\sz\pm0.03$   &   $\textbf{0.068}\sz\pm0.03$  &   $0.242\sz\pm0.04$   &  $0.176\sz\pm0.03$  & $0.168\sz\pm0.04$ \\
            \emph{Facebook2} &  $16$  &  $1.766\sz\pm0.53$&  $0.686\sz\pm0.41$   &   $\textbf{0.667}\sz\pm0.40$  &   $2.591\sz\pm0.57$   &  $1.910\sz\pm0.48$  & $1.555\sz\pm0.65$ \\	
            \emph{Milan}     &  $16$  &  $\textbf{0.015}\sz\pm0.01$&  $0.021\sz\pm0.01$   &   $0.021\sz\pm0.01$  &   $0.021\sz\pm0.00$   &  $0.116\sz\pm0.05$  & $0.127\sz\pm0.05$ \\
            \emph{T-Drive}   &  $16$  &  $0.038\sz\pm0.04$&  $0.027\sz\pm0.02$   &   $\textbf{0.024}\sz\pm0.01$  &   $0.252\sz\pm0.09$   &  $0.572\sz\pm0.26$  & $0.615\sz\pm0.26$ \\\cmidrule{1-8}
            \emph{Facebook1} &  $64$  &  $0.188\sz\pm0.04$&  $0.071\sz\pm0.03$   &   $\textbf{0.069}\sz\pm0.03$  &   $0.249\sz\pm0.04$   &  $0.178\sz\pm0.03$  & $0.168\sz\pm0.04$ \\
            \emph{Facebook2} &  $64$  &  $1.805\sz\pm0.52$&  $0.699\sz\pm0.41$   &   $\textbf{0.691}\sz\pm0.40$  &   $2.626\sz\pm0.57$   &  $1.910\sz\pm0.48$  & $1.556\sz\pm0.65$ \\
            \emph{Milan}     &  $64$  &  $\textbf{0.023}\sz\pm0.01$&  $0.025\sz\pm0.01$   &   $0.026\sz\pm0.01$  &   $0.026\sz\pm0.01$   &  $0.116\sz\pm0.05$  & $0.128\sz\pm0.05$ \\
            \emph{T-Drive}   &  $64$  &  $0.053\sz\pm0.04$&  $0.034\sz\pm0.02$   &   $\textbf{0.030}\sz\pm0.01$  &   $0.270\sz\pm0.10$   &  $0.576\sz\pm0.26$  & $0.617\sz\pm0.26$ \\
            \bottomrule
        \end{tabular}
    }
\end{table}
\begin{table}[ht]
    \centering
    \caption{Average candidate set sizes and standard deviation over 100 queries.}  
    \label{table:candidate_sets_sizes}
    \resizebox{1\linewidth}{!}{ 	\renewcommand{\arraystretch}{0.8}
        \hspace{-3mm}
        \begin{tabular}{lccccc}\toprule
            \multirow{5}{1cm}{\vspace*{4pt}\textbf{}\vspace*{4pt}}&\multicolumn{5}{c}{\textbf{Algorithm}}\\
            \cmidrule{2-6}
            \textbf{Data~set}       & \textsc{Tree} & \textsc{Pivot}  &  \textsc{Gbase} &  \textsc{Gshq} &  \textsc{Gshqt}\\ \midrule
            \emph{Facebook1}        & $ 256.0\sz\pm121.1$  &  $256.1\sz\pm121.1$  &  $891.7\sz\pm15.3$    &  $864.5\sz\pm58.5$     & $831.7\sz\pm182.9$      \\
            \emph{Facebook2}        & $2709.4\sz\pm1376.4$ & $2709.3\sz\pm1376.5$ &  $9820.9\sz\pm133.0$  &  $9408.4\sz\pm1045.5$  & $7729.9\sz\pm2739.0$    \\	
            \emph{Milan}            & $3120.9\sz\pm2469.8$ & $5178.0\sz\pm2723.8$ &  $4143.4\sz\pm1158.9$ &  $9311.6\sz\pm541.5$   & $9244.9\sz\pm757.8$    \\
            \emph{T-Drive}          & $1255.7\sz\pm1070.0$ & $1753.8\sz\pm1350.4$ &  $9873.2\sz\pm734.8$  &  $10231.0\sz\pm25.2$   & $10232.5\sz\pm25.0$       \\
            \bottomrule
        \end{tabular}
    }
\end{table}

\medskip
\noindent\textbf{Q2: } We selected $100$ trajectories randomly from the data sets as queries. The query interval is set to $\mathcal{I}(Q)$. We ran the algorithms for $k\in\{1,4,16,64\}$.
\Cref{table:treshold_radius} shows the threshold radii that we used for the pivot-based filtering, and \Cref{table:query_times} shows the average running times for querying a trajectory from the data set.
First, note that the query times of our exact approach (\textsc{Exact}) are lower than the query times of \textsc{Gbase} for all data sets. For the \emph{Facebook2} and the \emph{T-Drive} instances \textsc{Gbase} is up 0.2 seconds slower.
We will see later that \textsc{Gbase}, in contrast to our exact approach, does not always find the optimal solution set. Both of our index structures lead to accelerated query times compared to the exact approach for all data sets but the \emph{Milan} data set. The largest speed-up of about three to four is achieved for \textsc{Pivot} on the \emph{Facebook} instances. Note that almost always, the two heuristics \textsc{Gshq} and \textsc{Gshqt} are much slower in answering the queries. For the \emph{Milan} and \emph{T-Drive} instances, they are even slower than our exact approach. The reason is that they often have large candidate sets, see \Cref{table:candidate_sets_sizes}.
\textsc{Tree} is on-par with \textsc{Pivot} and has, in most cases, only a little higher running times beside the more complex data structure. 
The candidate set sizes of \textsc{Tree} and \textsc{Pivot} are similar for the \emph{Facebook} data sets, see \Cref{table:candidate_sets_sizes}. For the \emph{Milan} data set, \textsc{Tree} returns a smaller candidate set and has a slightly better running time for $k=1$ and $k=64$ compared to \textsc{Pivot}.
The full potential of the \textsc{Tree} index does not come to play for the other data sets due to the temporal distribution of the trajectories, and or the limited size.  
We suspect that the high running time of \textsc{Gbase} for \emph{Facebook2} and $k=1$ is an outlier and the result of the very high memory usage of the algorithm.
\Cref{fig:runningtimesk64} shows the average running times for $k=64$. %
\begin{figure}
    \centering
    \includegraphics[width=1\linewidth]{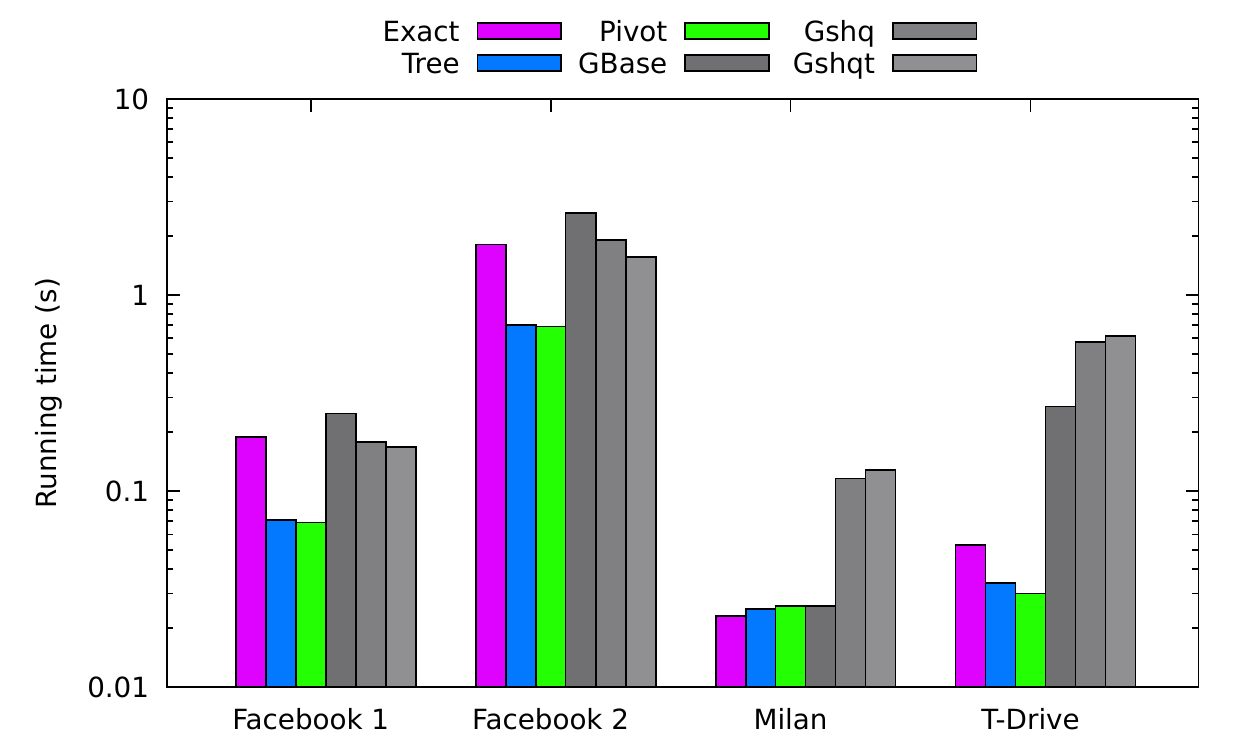}
    \caption{The average running times for $k=64$ over $100$ queries on a logarithmic scale.}
    \label{fig:runningtimesk64}
\end{figure}
\begin{table}[ht]
    \centering
    \caption{The SSR results are the averages and standard deviations over $100$ queries.}  
    \label{table:SSR_resultt}
    \resizebox{1.0\linewidth}{!}{ 	\renewcommand{\arraystretch}{0.8}
        \hspace{-3mm}
        \begin{tabular}{lcccccc}\toprule
            \multirow{5}{1cm}{\vspace*{4pt}\textbf{}\vspace*{4pt}}&\multicolumn{6}{c}{\textbf{Algorithm}}\\
            \cmidrule{3-7}
            \textbf{Data set}            & $k$ & \textsc{Tree} & \textsc{Pivot}  &  \textsc{Gbase} &  \textsc{Gshq} &  \textsc{Gshqt}\\ \midrule
            \emph{Facebook1} &  $1$   & $0.99\sz\pm0.01$ & $0.99\sz\pm0.01$ & $1.00\sz\pm0.00$ & $1.00\sz\pm0.00$ & $0.95\sz\pm0.19$ \\
            \emph{Facebook2} &  $1$   & $0.99\sz\pm0.00$ & $0.99\sz\pm0.00$ & $1.00\sz\pm0.00$ & $1.00\sz\pm0.00$ & $0.93\sz\pm0.23$ \\
            \emph{Milan}     &  $1$   & $0.95\sz\pm0.18$ & $0.96\sz\pm0.16$ & $1.00\sz\pm0.00$ & $1.00\sz\pm0.00$ & $1.00\sz\pm0.00$ \\
            \emph{T-Drive}   &  $1$   & $0.99\sz\pm0.02$ & $0.99\sz\pm0.02$ & $1.00\sz\pm0.00$ & $1.00\sz\pm0.00$ & $1.00\sz\pm0.00$ \\	\cmidrule{1-7}	
            \emph{Facebook1} &  $4$   & $0.99\sz\pm0.01$ & $0.99\sz\pm0.01$ & $1.00\sz\pm0.00$ & $1.00\sz\pm0.00$ & $0.95\sz\pm0.19$ \\
            \emph{Facebook2} &  $4$   & $0.99\sz\pm0.00$ & $0.99\sz\pm0.00$ & $1.00\sz\pm0.00$ & $1.00\sz\pm0.00$ & $0.93\sz\pm0.23$ \\	
            \emph{Milan}     &  $4$   & $0.90\sz\pm0.22$ & $0.96\sz\pm0.17$ & $0.99\sz\pm0.00$ & $1.00\sz\pm0.00$ & $0.99\sz\pm0.01$ \\
            \emph{T-Drive}   &  $4$   & $0.99\sz\pm0.03$ & $0.99\sz\pm0.03$ & $0.99\sz\pm0.00$ & $1.00\sz\pm0.00$ & $1.00\sz\pm0.00$ \\	\cmidrule{1-7}		
            \emph{Facebook1} &  $16$  & $0.97\sz\pm0.03$ & $0.97\sz\pm0.06$ & $1.00\sz\pm0.00$ & $1.00\sz\pm0.00$ & $0.95\sz\pm0.19$ \\	
            \emph{Facebook2} &  $16$  & $0.99\sz\pm0.00$ & $0.99\sz\pm0.00$ & $1.00\sz\pm0.00$ & $1.00\sz\pm0.00$ & $0.93\sz\pm0.23$ \\
            \emph{Milan}     &  $16$  & $0.81\sz\pm0.28$ & $0.91\sz\pm0.21$ & $0.99\sz\pm0.02$ & $1.00\sz\pm0.00$ & $0.99\sz\pm0.06$ \\
            \emph{T-Drive}   &  $16$  & $0.95\sz\pm0.09$ & $0.97\sz\pm0.09$ & $0.99\sz\pm0.00$ & $0.99\sz\pm0.00$ & $1.00\sz\pm0.00$ \\	\cmidrule{1-7}	
            \emph{Facebook1} &  $64$  & $0.91\sz\pm0.16$ & $0.91\sz\pm0.16$ & $1.00\sz\pm0.00$ & $1.00\sz\pm0.00$ & $0.95\sz\pm0.19$ \\
            \emph{Facebook2} &  $64$  & $0.99\sz\pm0.01$ & $0.99\sz\pm0.01$ & $1.00\sz\pm0.00$ & $1.00\sz\pm0.00$ & $0.93\sz\pm0.23$ \\
            \emph{Milan}     &  $64$  & $0.70\sz\pm0.30$ & $0.84\sz\pm0.23$ & $0.96\sz\pm0.04$ & $0.99\sz\pm0.00$ & $0.98\sz\pm0.07$ \\
            \emph{T-Drive}   &  $64$  & $0.87\sz\pm0.17$ & $0.91\sz\pm0.18$ & $0.99\sz\pm0.00$ & $0.99\sz\pm0.00$ & $1.00\sz\pm0.00$ \\	
            \bottomrule
        \end{tabular}
    }	
\end{table}
\begin{table}
    \centering
    \caption{Candidate set sizes $|\mathcal{C}|$, running times in $s$ and SSR for varying number of pivot elements $h$ and radii $r$ for the \emph{T-Drive} data set. We report the average and the standard deviation over 100 queries.}  
    \label{table:r_h_variation}
    \resizebox{1\linewidth}{!}{ 	\renewcommand{\arraystretch}{1}
        \hspace{-3mm}
        \begin{tabular}{lcc@{\hspace{4mm}}ccc@{\hspace{6mm}}ccc}\toprule
            \multirow{5}{0mm}{\vspace*{4pt}\textbf{}\vspace*{4pt}}&&&\multicolumn{3}{c}{$r=0.1$}&\multicolumn{3}{c}{$r=0.4$}\\
            \cmidrule{4-9}
            &$h$  &  $k$ & $size$ & $time$  & \emph{SSR} & $size$ &  $time$ & {SSR}\\ \midrule
            \multirow{4}{*}{\rotatebox{90}{\textsc{Tree}}}    
            & 64  &  1   & $354.7\sz\pm426.7$   &  $\textbf{0.006}\sz\pm0.00$  &  $\textbf{0.98}\sz\pm0.05$    &  $3392.8\sz\pm1931.4$  &  $0.053\sz\pm0.03$  &  $1.00\sz\pm0.00$  \\
            &128  &  1   & $247.6\sz\pm362.3$   &  $0.005 \sz\pm0.00$  &  $0.97\sz\pm0.07$    &  $2329.0\sz\pm1647.6$  &  $0.035\sz\pm0.02$  &  $0.99\sz\pm0.00$  \\	
            & 64  &  64  & $354.7\sz\pm426.7$   &  $0.012\sz\pm0.01$  &  $\textcolor{red}{0.61}\sz\pm0.30$    &  $3392.8\sz\pm1931.4$  &  $0.076\sz\pm0.04$  &  $0.96\sz\pm0.09$  \\
            &128  &  64  & $247.6\sz\pm362.2$   &  $0.009\sz\pm0.01$  &  $\textcolor{red}{0.43}\sz\pm0.35$    &  $2329.0\sz\pm1647.6$  &  $\textbf{0.055}\sz\pm0.03$  &  $\textbf{0.93}\sz\pm0.12$  \\\cmidrule{1-9}
            \multirow{4}{*}{\rotatebox{90}{\textsc{Pivot}}}   
            &64   &  1   & $337.3\sz\pm431.9$   &  $\textbf{0.006}\sz\pm0.00$  &  $\textbf{0.98}\sz\pm0.05$    &  $3391.5\sz\pm1953.8$ &  $0.031\sz\pm0.01$    &  $1.00\sz\pm0.00$  \\
            &128  &  1   & $230.0\sz\pm365.3$   &  $0.005\sz\pm0.00$  &  $0.97\sz\pm0.07$    &  $2369.6\sz\pm1666.5$ &  $0.023\sz\pm0.01$    &  $0.99\sz\pm0.00$  \\	
            &64   &  64  & $337.3\sz\pm431.9$   &  $0.009\sz\pm0.00$  &  $\textcolor{red}{0.61}\sz\pm0.30$    &  $3391.5\sz\pm1953.8$ &  $0.051\sz\pm0.02$    &  $0.96\sz\pm0.09$  \\
            &128  &  64  & $230.0\sz\pm365.3$   &  $0.007\sz\pm0.00$  &  $\textcolor{red}{0.43}\sz\pm0.35$    &  $2369.6\sz\pm1666.5$ &  $\textbf{0.038}\sz\pm0.02$    &  $\textbf{0.93}\sz\pm0.12$  \\
            \bottomrule
        \end{tabular}
    }
\end{table}
\begin{figure}[htb]
    \centering
    \includegraphics[width=1.0\linewidth]{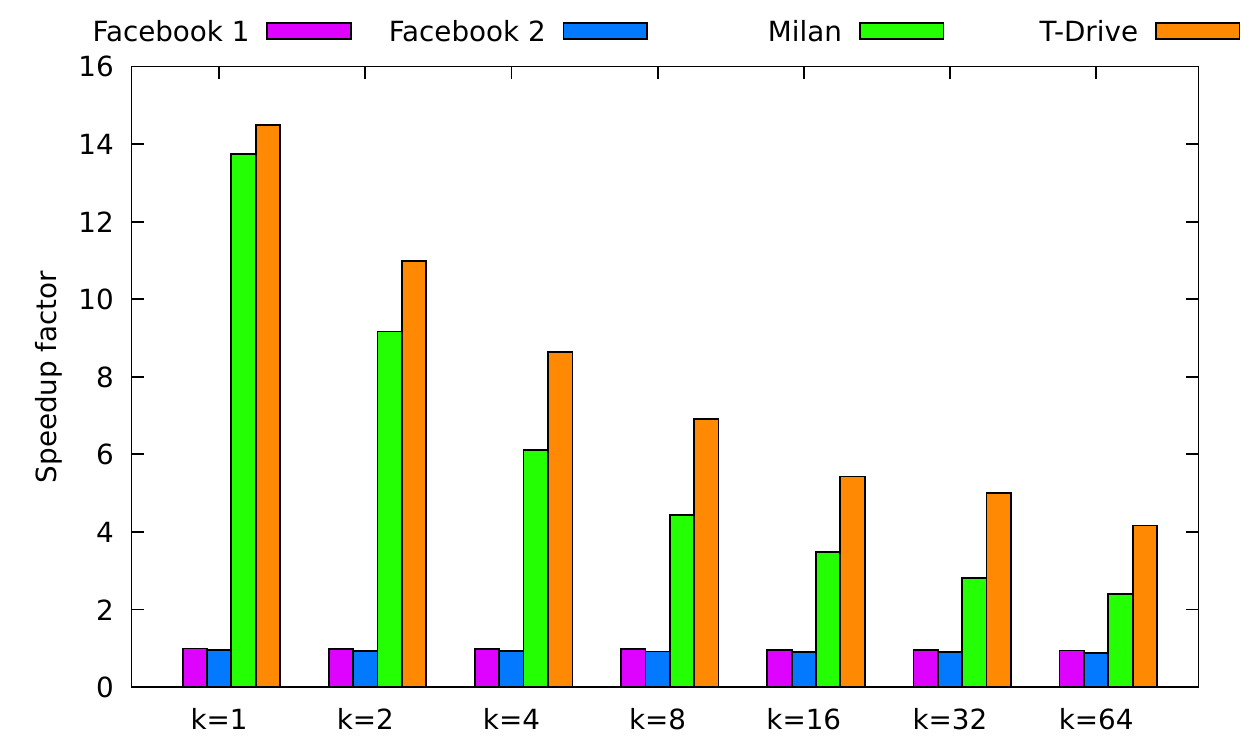}
    \caption{Average speed up by using upper bounding during the calculation of the similarity over $100$ queries}
    \label{fig:blupspeedup}
\end{figure}
\begin{figure*}[t]
    \includegraphics[width=0.24\textwidth]{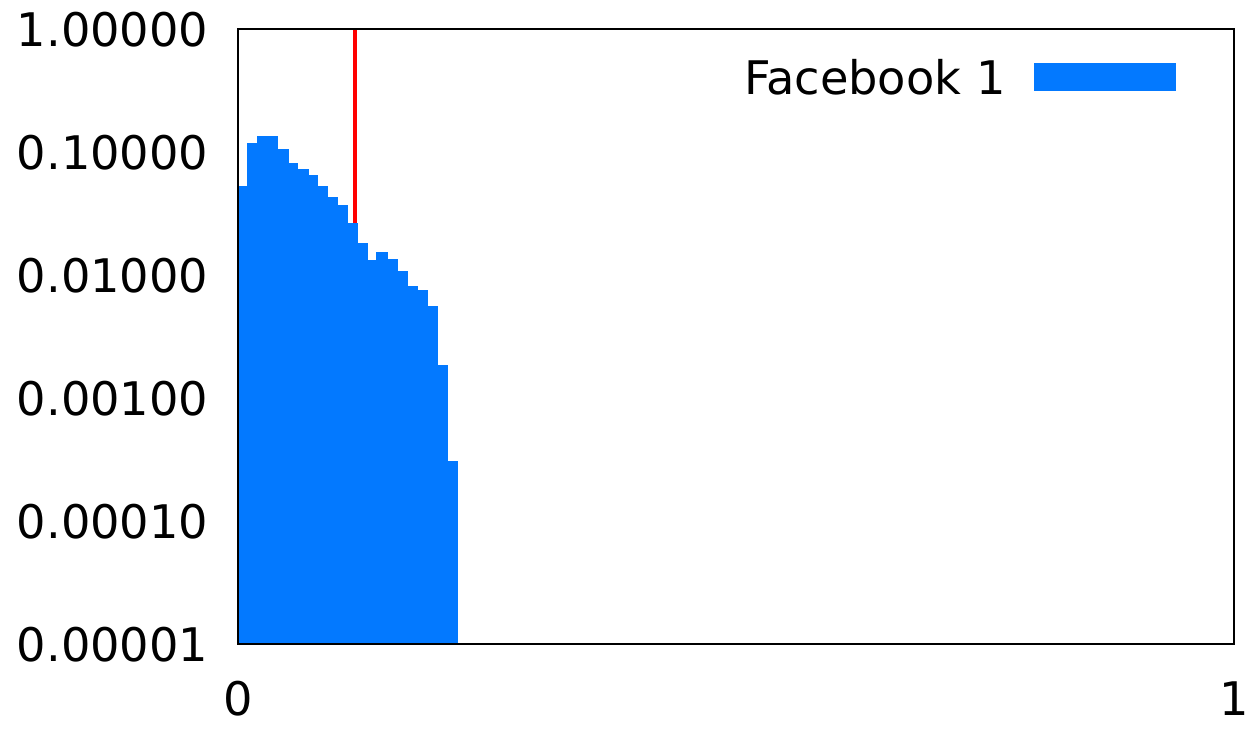}
    \hfill
    \includegraphics[width=0.24\textwidth]{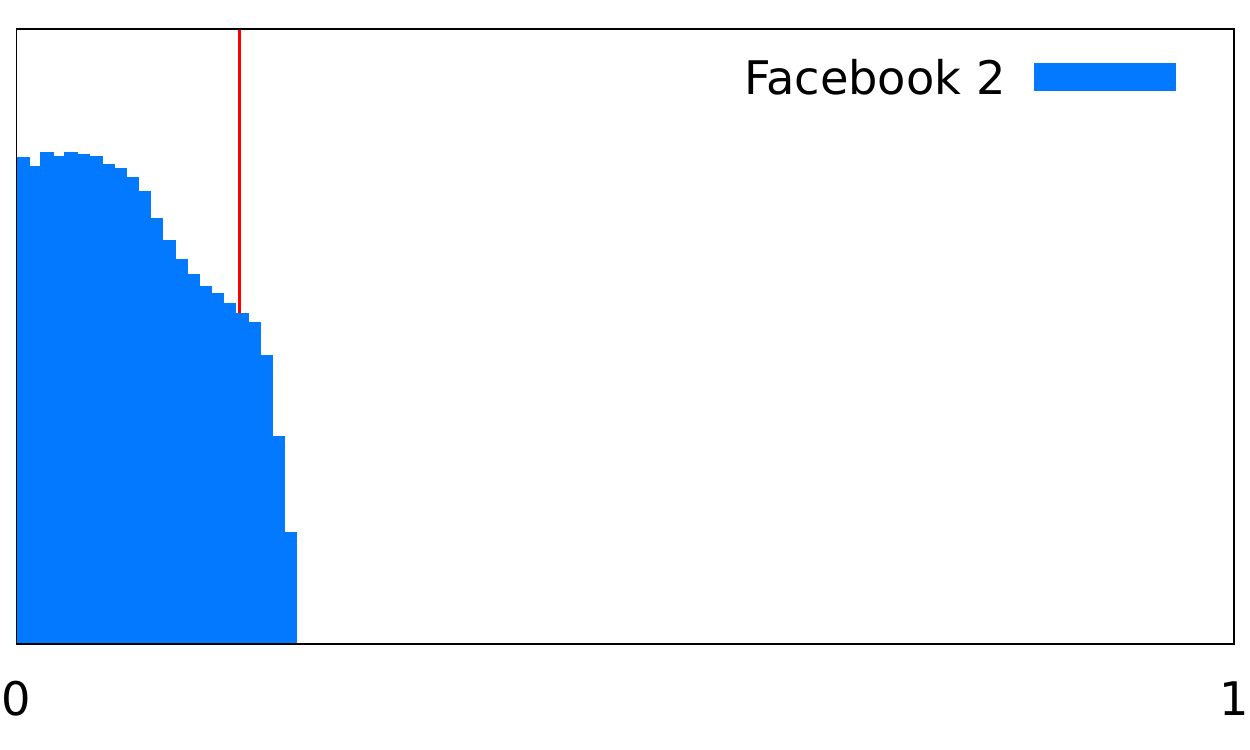}
    \hfill
    \includegraphics[width=0.24\textwidth]{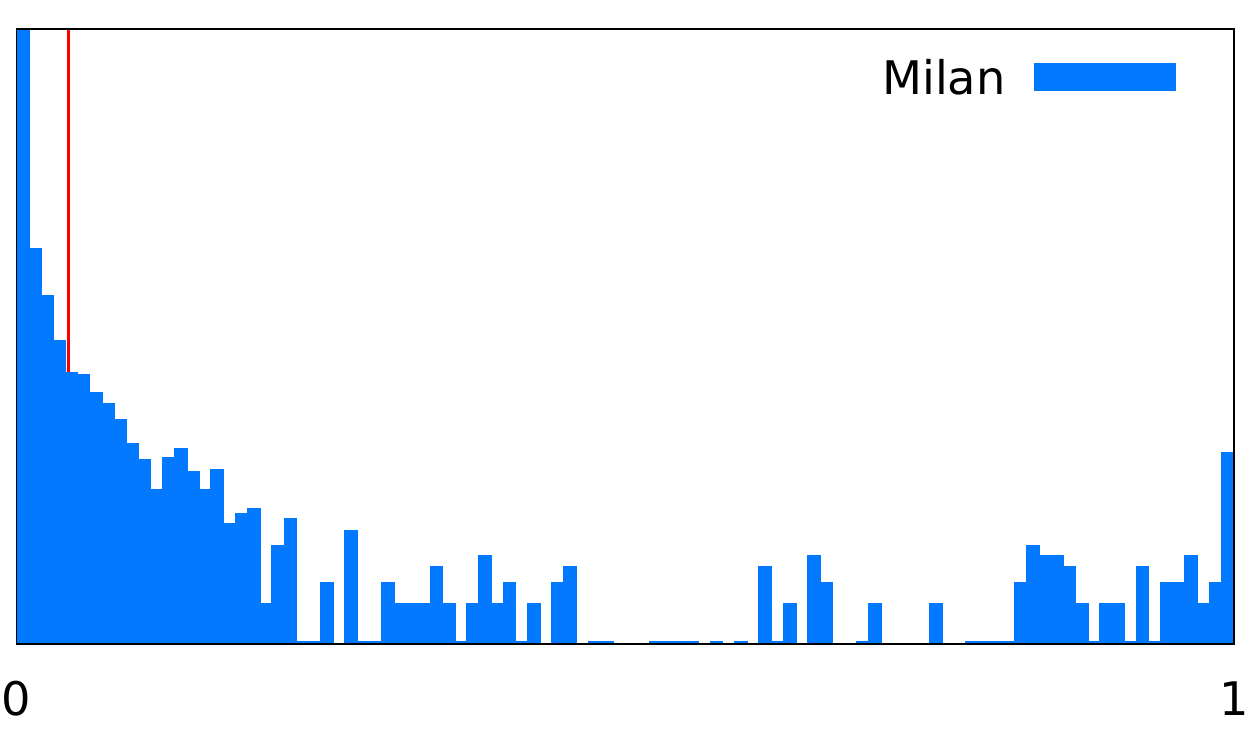}
    \hfill
    \includegraphics[width=0.24\textwidth]{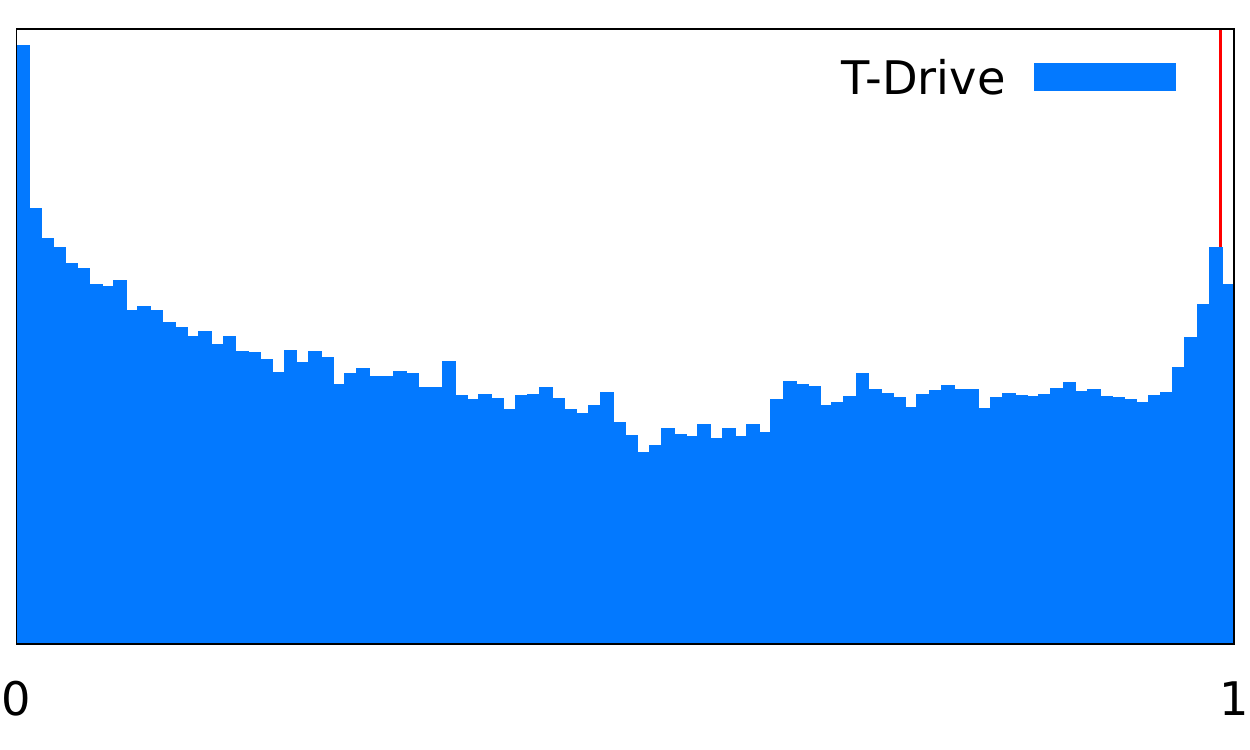}
    \caption{Distribution of similarities: The $x$-axis ranges from $0$ to $1$. The $y$-axis shows the fraction of input-query pairs that have this similarity on a logarithmic scale. The red line highlights the 10th percentile for \emph{Facebook1} and the 1st percentile for the other data sets. All the similarities of trajectories found by our queries lie to the right of the displayed red lines.}
    \label{fig:distances}
\end{figure*}

\medskip
\noindent\textbf{Q3: } 
In order to evaluate the quality of our query results, we use the \emph{similarity score ratio} (SSR) defined in~\cite{grossi2020finding}. The SSR of two sets $\mathcal{T}_1$ and $\mathcal{T}_2$ of trajectories with respect to a query is defined as $SSR(\mathcal{T}_1, \mathcal{T}_2, (Q,s))=\frac{\sum_{T_1\in \mathcal{T}_1}Sim(Q,T_1,s)}{\sum_{T_2\in \mathcal{T}_2}Sim(Q,T_2,s)}$.
We compare the results of the indices to the results of the exact algorithm \textsc{Exact}. \Cref{table:SSR_resultt} shows the average SSR values and the standard deviations over 100 queries.

First we observe that as expected (see section~\ref{sec:grossi}) the baseline \textsc{Gbase} \cite{grossi2020finding} has not always found the optimal solution set. The SSR score takes values below one for the \emph{Milan} and the \emph{T-Drive} data sets. However, for the optimal solution, the SSR value should be one.
With increasing value of $k$ the SSR value for our \textsc{Tree} algorithm decreases from 0.99 for $k=1$ to 0.70 for $k=64$.
However, for our \textsc{Pivot} approach the decrease is less strong; the SSR score is always above 0.91 for the instances
\emph{Facebook1}, \emph{Facebook2}, and \emph{T-Drive}. For the \emph{Milan} instance, our heuristics do not behave very well for large $k$.
Here, the SSR score for \textsc{Tree} takes a value of $0.7$ for $k=64$.
The reason for this low value is the small value of $r$ chosen in our experiments.
However, a larger value of $r$ will lead to even higher running time compared to the exact computations, which is already faster. 
This is because of the length of the \emph{Milan} trajectories are relatively small (see~\Cref{table:datasets_stats2}). %
For the \emph{Facebook2} instances the SSR score of \textsc{Pivot} is always 0.99.
The values of the \textsc{Gshq} and \textsc{Gshqt} heuristics for \emph{Facebook1}, \emph{Facebook2}, and \emph{T-Drive} are always above 0.93 due to the usage of the large candidate sets (see \Cref{table:candidate_sets_sizes}).
However, remember that their query times take are much longer than that for \textsc{Tree}, \textsc{Pivot}, and even our exact computations.

\medskip
\noindent\textbf{Q4: } By increasing the number of pivot elements $h$, a larger number of trajectories may be excluded from the candidate set, since every pivot adds an additional filter. However, each additional pivot might lead to a higher number of false negatives, i.e., trajectories that are not part of the candidate set but are part of the optimal top-$k$ set.
For the \emph{T-Drive} data set, we ran \textsc{Tree} and \textsc{Pivot} with $h\in\{64,128\}$ and $r\in\{0.1,0.4\}$. 
For building the index, \textsc{Tree} took $1.41$ seconds and \textsc{Pivot} took $0.85$ seconds.
\Cref{table:r_h_variation} shows the effect on query times and quality. We compare these results to \Cref{table:SSR_resultt} and \Cref{table:query_times} (there, the value of $r$ was chosen as $0.2$ and $0.25$, respectively).

Lowering $h$ and increasing $r$, each increases the size of the candidate set. A larger candidate set may lead to better SSR values; however, it also increases the running time. Notice, for $k=1$ we can achieve faster running times with high SSR value by choosing a small radius $r=0.1$ compared to the results in \Cref{table:SSR_resultt}. On the other hand, for $k=64$, \textsc{Pivot} improves its SSR value compared to \Cref{table:SSR_resultt} by choosing $h=128$ pivots and radius $r=0.4$, while being faster than \textsc{Exact}.

\medskip
\noindent\textbf{Q5: }
To evaluate the speedup gained by the upper bounding technique, we computed the similarity for $100$ queries for $k=2^i$ with $0\leq i \leq 6$. For each $k$, we computed the top-$k$ results without indexing, with and without the upper bounding. \Cref{fig:blupspeedup} shows the speedup that is achieved by using the upper bounding technique.
The \emph{T-Drive} and \emph{Milan} data sets profit immensely with speedups between over $4$ and $14$, and $2$ and $13$, respectively.
The speedups decrease with increasing $k$. The reason is that there are often only a few trajectories with very high similarity. If the algorithm finds these early on during the processing of the query and if the value of $k$ is small, then the upper bounding is most effective. For larger $k$, the lowest of the top-$k$ similarities is closer to the non-top-$k$ similarities, and upper bounding, i.e., stopping the computation early, happens less often.
There is no speedup in the case of the \emph{Facebook} data sets. 
The reason is that the differences in the similarities between the query and the trajectories are small (see \Cref{fig:distances}). 
Moreover, due to the long trajectories (see \Cref{table:datasets_stats2}), the upper bounds have to be updated often, such that the upper bounding in total cannot speed up the query.
\section{Conclusion}
We studied computing the top-$k$ most similar trajectories in a graph to a given query trajectory.
For this, we proposed a new spatio-temporal similarity measure based on the work of Grossi et al.~\cite{grossi2020finding}. 
We derived a distance function from our new similarity function, which satisfies a triangle inequality under certain conditions.
That built the basis for our pivot-based filtering technique, which accelerates finding exact solutions of top-$k$ trajectory queries.
Furthermore, we suggested a tree-based temporal filtering method in combination with the pivot-based technique.
Both approaches strongly outperform the baselines for all data sets, but the \emph{Milan} data set.
Here, our new baseline algorithm that uses the upper bounding technique has the lowest running time. It is also the first exact algorithm for the top-$k$ trajectory problem,
as we showed that the baseline in~\cite{grossi2020finding} does not always find the exact solution.
\section*{Acknowledgments}
This work is funded by the Deutsche Forschungsgemeinschaft (DFG, German Research Foundation) under Germany's Excellence Strategy -- EXC-2047/1 -- 390685813.
\balance
\bibliographystyle{abbrv}
\bibliography{literature}

\end{document}